\documentclass[runningheads,a4paper]{llncs}

\usepackage{indentfirst}

\usepackage[utf8]{inputenc}
\usepackage[unicode]{hyperref}
\usepackage[english]{babel}

\usepackage{graphicx, url}
\usepackage{amssymb}
\usepackage{amsmath}
\usepackage{algorithm}
\usepackage[noend]{algpseudocode}

\renewcommand{\geq}{\geqslant}
\renewcommand{\leq}{\leqslant}
\newtheorem{Theorem}{Theorem}
\newtheorem{Lemma}{Lemma}
\newtheorem{col}{Corollary}

\begin{document}

\mainmatter

\title{On mixing in pairwise Markov random fields 
with application to social networks}

\titlerunning{On mixing in pairwise Markov random fields}


\author{Konstantin Avrachenkov\inst{1}, \ Lenar Iskhakov\inst{2}, \ Maksim Mironov\inst{2}}

\institute{Inria Sophia Antipolis, 
2004 Route des Lucioles, Sophia-Antipolis, France
\email{k.avrachenkov@sophia.inria.fr}
\and
Moscow Institute of Physics and Technology, Dolgoprudny, Russia
\email{maxim-m94@mail.ru, lenar-iskhakov@yandex.ru} }

\authorrunning{K. Avrachenkov, L. Iskhakov, 
M. Mironov}


\maketitle

\begin{abstract}
We consider pairwise Markov random fields which have
a number of important applications in statistical physics,
image processing and machine learning such as Ising model
and labeling problem to name a couple. Our own motivation
comes from the need to produce synthetic models for
social networks with attributes. First, we give conditions for 
rapid mixing of the associated Glauber dynamics and consider
interesting particular cases. Then, for pairwise Markov random 
fields with submodular energy functions we construct monotone perfect simulation. 
\end{abstract}

\section{Introduction}

Pairwise Markov random fields or Markov random fields with nonzero potential functions only for cliques of size two have a large number of applications in statistical physics, image processing and machine learning. Let us mention just a few very important particular cases and
applications. Ising \cite{Ising25}, Potts \cite{Potts52} 
and Solid-on-Solid (SOS) \cite{MS91,RS06} models
are the basic models in statistical physics. 
Metric Markov random fields and the generalized Potts
model are very successfully applied in image 
processing \cite{BVZ98,BVZ01,Setal08}. 
Pairwise Markov random fields are also
extensively used in the study of classification and
labeling problems, see e.g. \cite{BBM04,CDI98,KT02}.

Our own motivation to study pairwise Markov random fields comes from the need to model the distribution of attributes
in social networks such as age, gender, interests. 
The fact that friends or acquaintances in social networks share common characteristics is widely observed in real networks and is referred to as homophily.
The property of homophily implies that we expect that
the more clustered social network members are, the more likely
they are to share same attribute.
Nowadays social networks are intensively researched 
by both sociologists and computer scientists. However, 
if one wants to check some hypotheses about social networks or to test some algorithm such as a sampling method, 
the researchers need a lot of social network examples 
to consider and to test. In \cite{ANT16} a model of synthetic 
social network with attributes has been proposed to
test subsampling chain-referral methods on many network instances
with various properties. The synthetic network model
of \cite{ANT16} is similar in spirit to the SOS model and
well represents the distribution of ordinal attributes 
such as age.
Here we study much more general model which could
be used to model ordinal as well as non-ordinal
attributes' distribution in social networks. Of course,
we hope that the results will also be of interest to
researchers from statistical physics and machine learning
communities.

Specifically, in the present work we consider a general
pairwise Markov random field and provide conditions
for rapid mixing of the associated Glauber dynamics.
Rapid mixing guarantees that we can quickly generate
many configurations of attributes corresponding to
a given Gibbs distribution or energy function. In the
important particular case of submodular energy functions,
we go a step further and construct a perfect simulation
which samples quickly without bias from the target
distribution. Our results significantly generalize
the corresponding results for the Ising model, see 
e.g. \cite{LevinMCaMT}. The proof in \cite{LevinMCaMT} 
relies on the particular size and values of the
interaction matrix.

Finally, we would like to note that even though our
model has some common features with the exponential random 
graph model (see e.g., \cite{RPKL07}), 
there are important differences between
these two models. The exponential random graph model
generates the graph, whereas our model assumes that
the graph is given and generates a configuration of
attributes over the graph.

\section{Model}

Let a graph $G = (V,E),$ $|V| = n$, be given. In addition, each
vertex $v$ has an attribute which takes a value from the finite
set $M = \{1,...,m\}$. We denote by $\sigma \in \Omega=M^n$ a configuration, where each vertex $v \in V$ takes its own certain value 
$\sigma(v) \in M$ of the attribute. In the present work we restrict
ourselves to the model with one attribute. Now we introduce symmetric {\it interaction} matrix $\mathbb{V}$ of size $m\times m$, and say, that the energy of configuration $\sigma$ is given by
$$
\varepsilon(\sigma) = \sum_{\{v_1, v_2\} \in E}\mathbb{V}(\sigma(v_1), \sigma(v_2)).
$$
Let us call $|\mathbb{V}|$ the maximum absolute value of matrix $\mathbb{V}$ elements. Next we consider \textit{Gibbs distribution} with respect to the introduced energy:
$$
\pi^*(\sigma) = \frac{e^{-\beta \varepsilon(\sigma)}}{\sum \limits_{\tau \in M^G}e^{-\beta \varepsilon(\tau)}} = Z^{-1}(\beta)e^{-\beta\varepsilon(\sigma)},
$$
where $\beta = \frac{1}{T}$ is some parameter, the inverse temperature 
of the system, and $Z(\beta)$ is the normalizing constant or, in
statistical physics terminology, the partition function. 
This distribution describes the {\it pairwise Markov
random field} over graph $G$. We shall also refer to this distribution 
as \textit{network attribute distribution}.

We would like to sample configurations from the distribution $\pi(\sigma)$ to test
various algorithms on a series of network realisations. 
However, the main
problem is that the probability space is enormous and it is impossible 
to sample from Gibbs distribution without additional techniques.
One such technique is Glauber dynamics, described just below and
another technique is monotone perfect simulations described in
detail in Section~5.

Let $\mathcal{N}(v)$ be the set of neighbours of vertex $v$.
Then, we define the \textit{local energy} $\varepsilon_i(\sigma, v)$
for vertex $v$ and value $i$ in configuration $\sigma$ as follows: 
$$
\varepsilon_i(\sigma, v) = \sum_{u \in \mathcal{N}(v)}\mathbb{V}(i, \sigma(u)). 
$$
This formula calculates energy in the neighbourhood of $v$ provided that the value of the attribute for $v$ was updated to $i$.
Then, we call the \textit{local distribution} for vertex $v$ in configuration $\sigma$ the probability distribution on set 
$\{1,2,\ ...\ ,m\}$ with respect to the local energy:
$$ p_i(\sigma, v) = {\mathbb{P}}(\sigma(v) \to i) := \frac{e^{-\beta \varepsilon_i(\sigma, v)}}{\sum\limits_{k \in M}e^{-\beta \varepsilon_k(\sigma,v)}} = Z^{-1}(\sigma,v, \beta)\cdot e^{-\beta \varepsilon_i(\sigma, v)},$$
which is the probability to update value in $v$ to $i$.

The Glauber dynamics is defined as follows:
\begin{enumerate}
\item Choose arbitrary starting distribution $\pi^0$ and 
then choose values for vertices according to $\pi^0$;
\item Choose uniformly random vertex $v$;
\item Update value for $v$ according to the local distribution;
\item Go to step 2.
\end{enumerate}

Let us denote by $\mathcal{X} = \{ X_t, t\geq 0 \}$ the Markov chain 
associated with the Glauber dynamics, with starting distribution $\pi^0$ and transition matrix $P = \{P_{\sigma, \tau}\}_{\sigma, \tau \in \Omega},$ 
$\ P_{\sigma,\tau} = \mathbb{P}\{X_{t+1} = \tau|X_{t} = \sigma\}$, which is associated with steps 2-3. If steps 2-3 are repeated $t$ times, $\pi^t$ will stand for the distribution on space of configurations at time moment $t$.
Sometimes we shall also use $P_{\sigma}^t(\cdot)$ to denote the probability distribution of $\mathcal{X}$ on $\Omega$ 
at time moment $t$ to emphasize that $\mathcal{X}$ starts 
from certain configuration $\sigma$.

Before we proceed further, let us notice that the introduced model implies some well-known particular cases. For example, 
$$\mathbb{V} = \begin{pmatrix}
1 & 0 & \cdots & 0 \\
0 & 1 & \cdots & 0 \\        
\vdots & \vdots & \ddots & \vdots \\
0 & 0 & \cdots & 1
\end{pmatrix}$$
corresponds to the Potts model. If $m=2$, then the Potts model becomes
the Ising model. If now we take $\mathbb{V}(i,j)=f(|i-j|)$ with some convex function
$f(\cdot)$, we obtain the metric Markov random field model extensively
used in image processing. In \cite{ANT16}, the Markov random field
with quadratic $f(\cdot)$ was used to model social networks with
ordinal attributes. The case $\mathbb{V}(i,j)=|i-j|$ corresponds to the SOS model. 


\section{Preliminaries}

Here we give several well-known results, which we will use in sequel.

It is well-known, see e.g., \cite{Bremaud} and \cite{LevinMCaMT}, that the Markov chain $\mathcal{X}$ corresponding to the Glauber dynamics is reversible with the stationary distribution $\pi^*$.
\begin{Lemma}
Markov chain $\mathcal{X}$ is time-reversible with the stationary distribution given by
$\pi^*(\sigma)=Z^{-1}(\beta) e^{-\beta \varepsilon(\sigma)}$. 
In other words,
$$ 
\pi^*(\sigma)\cdot P_{\sigma, \tau}  = \pi^*(\tau)\cdot P_{\tau, \sigma},
$$
for all $\sigma, \tau \in \Omega.$
\end{Lemma}

For two distributions $\pi_1, \pi_2$ on state space $\Omega$ we define the \textit{total variation} distance between them as
$$ 
|| \pi_1 - \pi_2 ||_{TV} = \frac{1}{2}\sum_{\sigma \in \Omega}|\pi_1(\sigma) - \pi_2(\sigma)|. 
$$
Let $\mu$ and $\nu$ be two distributions on the same state space $\Omega$. Pair of random variables $(X_{\mu},X_{\nu})$ forms \textit{coupling}, if it is distributed such that marginal distribution of $X_{\mu}$ is $\mu$ and marginal distribution of $X_{\nu}$ is $\nu$. 
The main motivation for introducing such term is the following lemma \cite{Bremaud}.

\begin{Lemma}
Let $\nu$ and $\mu$ be two probability distributions on $\Omega$. Then
$$ || \mu - \nu ||_{TV} = \inf \{\ {\mathbb{P}}(X_{\mu} \neq X_{\nu})\ |\ (X_{\mu},X_{\nu})\ is\ a\ coupling\ of\ \mu\ and\ \nu\ \}. $$
\end{Lemma}
This lemma is very useful, because a comparison between distributions is reduced to comparison between random variables.

Here is one more lemma, which shows how the total variation distance from the stationary distribution can be estimated \cite{Bremaud,LevinMCaMT}.
\begin{Lemma}
Let $\sigma$ and $\tau$ be initial configurations from state space $\Omega$. Then
$$ || \pi^t - \pi^* ||_{TV} \leq \max \limits_{\sigma,\tau \in \Omega} || P_{\sigma}^t(\cdot) - P_{\tau}^t(\cdot) ||_{TV}. $$
\end{Lemma}
Now we introduce metric on configuration space $\Omega$.
Let $\rho(\cdot, \cdot)$ by definition be equal to
$$ \rho(\sigma, \tau) = \sum_{v \in V}|\sigma(v) - \tau(v)|.$$
\begin{Lemma}
Let $\alpha$ be such that for every two neighbor configurations $\sigma, \tau \\ (\rho(\sigma, \tau) = 1)$ corresponding random values $X_{\sigma}^1$ and $X_{\tau}^1$ satisfy an inequality
$$ {\sf E}\rho(X_{\sigma}^1, X_{\tau}^1) \leq e^{-\alpha}.$$ 
Then
$$\forall t\in \mathbb{N},\ \forall \sigma, \tau \in \Omega \to {\sf E}(\rho(X_{\sigma}^t, X_{\tau}^t)) \leq {\sf diam}(\Omega)\cdot e^{-\alpha t}.$$
\end{Lemma}
Lemma 4 shows how the introduced property can be generalized from neighbor configurations to the whole space $\Omega$ for an arbitrary time moment.

For some $\varepsilon > 0$, the {\it mixing time} is defined
as follows:
$$
t_{mix}(\varepsilon) = 
\min(t \in \mathbb{N}\ | \ ||\pi^t - \pi||_{TV} < \varepsilon).
$$
Next lemma is based on Lemma~4 and it provides an upped bound for the mixing time with respect to $\alpha$.
\begin{Lemma}
Suppose $\alpha > 0$ is such that ${\sf E}(\rho(X_{\sigma}^1, X_{\tau}^1)) \leq e^{-\alpha}$ for all neighbour configurations  $\sigma, \tau$. Then
$$t_{mix} \leq \left \lceil \frac{1}{\alpha}[\ln({\sf diam}(\Omega)) + \ln(1/\varepsilon)] \right \rceil.$$
\end{Lemma}
Lemmas 4 and 5 are borrowed from \cite{LevinMCaMT}. 
Actually, for our following results it would be enough to refer only to Lemma~5. But we mention here intermediate steps to help a reader to better understand the proof of our main result.

\section{Main results}

We can now formulate the main result of this article which says that
under certain conditions the Glauber dynamics corresponding to the
general pairwise Markov random fields mixes rapidly.

\begin{Theorem}
Let $\triangle$ be the maximum degree of graph $G = (V,E),\ |V| = n$ and $\mathbb{V}$ be the interaction matrix. 
Let also $\beta$ be the inverse temperature and $M = \{1,2,\ ...\ ,m\}$ be the set of attribute values. If
$$
\beta < \frac{1}{4|\mathbb{V}|}\ln \left(1 + \frac{1}{\triangle m}\right), 
$$
then
$$ 
t_{mix} \leq \left  \lceil \frac{n(\ln(n)+ \ln(m-1) + \ln(\frac{1}{\varepsilon}))}{1 - \triangle m(e^{4\beta |\mathbb{V}|} - 1)}\right \rceil. 
$$
\end{Theorem}

We would like to notice that independently from temperature the mixing time is at least of order $n \ln(n)$. It is so, because achieving stationary distribution by iterating means that every vertex of the graph has to be updated at least once. As $n$ grows
to infinity, we must do order $n\ln(n)$ Markov chain steps 
to make the probability of updating each vertex at least once tending to 1. More details on various lower bounds can be
found in \cite{LevinMCaMT}.

Before we proceed to prove the theorem, let us also notice that it claims that the upper bound is of order $n\log n$. 
The corresponding result for the Ising model has been
shown in e.g., \cite{LevinMCaMT}. The present extension
is not straightforward, since the proof in \cite{LevinMCaMT}
is based on the particular form of the interaction matrix
$\mathbb{V}$.

\begin{proof}
Let us choose two arbitrary configurations $\sigma$ and $\tau$ at time 0 and say that random vectors $X_{\sigma}^t$ and $X_{\tau}^t$ have distributions $P_{\sigma}^t(\cdot)$ and $P_{\tau}^t(\cdot)$, respectively. Then define $pref_k(\sigma, w), k\leq m$, as the {\it prefix sum} of probabilities to label $w$ with one of the first $k$ attribute values at the next step,
namely,
$$
pref_k(\sigma, w) = \sum_{i = 1}^k p_i(\sigma, w).
$$
Let us consider the following probability distribution of pair $(X_{\sigma}^t, X_{\tau}^t)$: first we uniformly at random choose a vertex $w$ to update (common for both configurations) and then we choose uniformly at random a value $U$ from $[0,1]$. Then we set new configurations $X_{\phi}^t(U,w), \phi \in \{\sigma, \tau \}$ at 
time $t$ by the relation
\begin{equation}
X_{\phi}^{t}(U,w) (\overline{w}) =
 \begin{cases}
   \phi(\overline{w}) &\overline{w} \neq w\\
    \min(k | pref_k(\phi,w) \geq U) &\overline{w} = w
 \end{cases}
 .
\end{equation}
where function $X_{\phi}^1:[0,1]\times V \rightarrow \Omega$ becomes a random vector, if $U$ and $w$ are random variables.


It is easy to see that distribution of pair $(X_{\sigma}^t(U,w), X_{\tau}^t(U,w))$ is coupling for $P_{\sigma}^t(\cdot)$ and $P_{\tau}^t(\cdot)$.

Then, we are going to find an $\alpha > 0$ from Lemma~4 for two neighbor configurations. Let $\sigma, \tau$ be two neighbor configurations with unique difference in vertex $v$, i.e., ${|\sigma(v)-\tau(v)|=1}$. Let also $w$ be a uniformly chosen random vertex. If $w = v$, then $$
\rho(X_{\sigma}^1(U,w), X_{\tau}^1(U,w)) = 0.
$$
If $w \notin \mathcal{N}(v) \cup \{v\}$, then
$$
\rho(X_{\sigma}^1(U,w), X_{\tau}^1(U,w)) = |\sigma(v) - \tau(v)| = 1.
$$
It is so, because in both cases local distributions for $w$ are the same for both configurations. And if $w \in \mathcal{N}(v)$, then 
$$
\rho(X_{\sigma}^1(U,w), X_{\tau}^1(U,w)) = |\sigma(v) - \tau(v)| + |X_{\sigma}^1(U,w)(w) - X_{\tau}^1(U,w)(w)|.
$$
According to probabilities of each case, we can write
\begin{equation}
{\sf E}\rho(X_{\sigma}^1(U,w), X_{\tau}^1(U,w)) = 1 - \frac{1}{n} + \frac{1}{n}\cdot \sum_{w \in \mathcal{N}(v)}{\sf E}| X_{\sigma}^1(U,w)(w) - X_{\tau}^1(U,w)(w) |.
\end{equation}
Thus, an upper bound for the sum in (2) is needed. The following lemma helps to achieve the result and is the key element of this work.
\begin{Lemma}
For arbitrary $\sigma, \tau \in \Omega$ and for all $w \in V$ the following equation holds
\begin{equation}
{\sf E}| X_{\sigma}^1(U,w)(w) - X_{\tau}^1(U,w)(w)| = \sum_{i = 1}^m| pref_{i}(\sigma, w) - pref_{i}(\tau ,w) |.
\end{equation}
\end{Lemma}

\begin{proof}
The expectation in (3) is based on uniform random variable $U$ distributed on $[0,1]$. Let us place on segment $[0,1]$ precisely $m$ red points that correspond to $pref_i(\sigma, w)$ and $m$ blue points that correspond to $pref_i(\tau, w)$, $1 \leq i \leq m$. Since $pref_m(\sigma, w) = pref_m(\tau, w) = 1$, we have $2m - 1$ disjoint (with no common internal points) subsegments with red or blue endpoints (some subsegments may have length 0), they form a set $\{l_k\}_{k=1}^{2m-1}$. Let subsegment $l_k$ have a value $h_{\sigma,k}$, if $h_{\sigma,k}$ satisfies $l_k \subset [pref_{h_{\sigma,k} - 1}(\sigma, v), pref_{h_{\sigma,k}}(\sigma, w)]$. Thus, by definition the mean of $|X_{\sigma}^1(U,w)(w) - X_{\tau}^1(U,w)(w)|$ is
$$ {\sf E}|X_{\sigma}^1(U,w)(w) - X_{\tau}^1(U,w)(w)| = \sum_{k = 1}^{2m-1}{\sf length}(l_k)\cdot |h_{\sigma, k} - h_{\tau, k}|.  $$
In other words, the length of $l_k$ appears in the expectation as many times as the difference between the values of the attribute for updates in $\sigma$ and $\tau$. Therefore, we now calculate the number of times that the length of each subsegment is added to the result in the right hand side of the above equality.
Towards this goal, for the moment let us fix $k$
and let $h_{\sigma, k} = a$, $h_{\tau, k} = b$ and without loss of generality $b \geq a$. Thus, the following series of inequalities hold
$$
\begin{cases}
pref_a(\sigma,w) \geq pref_a(\tau, w),\\
pref_{a+1}(\sigma,w) \geq pref_{a+1}(\tau, w),\\
...\\
pref_{b}(\sigma,w) \geq pref_b(\tau, w).
\end{cases}
$$
Let us identify terms  
$|pref_{i}(\sigma, w) - pref_{i}(\tau ,w)|$ in (3) which contain
the contribution from the subsegment $l_k$.
The length of $l_k$ is added for the first time in the right
hand side of (3) for $i=a$, because according to the definition of $a$ the minimum $i$ such that segment $[0, pref_{i}(\sigma, w)]$ contains $l_k$ is $i = a$, meantime $pref_a(\tau, w)$ does not contain this subsegment. Second time it is added for $i = a+1$ and so on, the last time it is added for $i = b-1$, which comes from definition of $b$. Hence, $l_k$ is added exactly $b-a$ times.
This establishes equivalence between the sums and completes the proof of the lemma. 
\hfill $\Box$
\end{proof}

Actually, this lemma will be used only for neighbor configurations $\sigma, \tau$, as it was mentioned before Lemma~6. Recall that Lemma~4 and then Lemma~5 give us an upper bound on the mixing time, but to apply them we need to obtain the corresponding inequalities on neighbour configurations. Therefore, we give a uniform upper bound for (3). For convenience we introduce
$$ S_i = \sum_{u \in \mathcal{N}(w)\setminus\{v\}}\mathbb{V}(i, \sigma(u)) = \sum_{u \in \mathcal{N}(w)\setminus\{v\}}\mathbb{V}(i, \tau(u)),$$
$$ a_i = \exp\left (-\beta\sum_{u\in \mathcal{N}(w)}\mathbb{V}(i, \sigma(u)) \right ) = \exp\left (-\beta (S_i + \mathbb{V}(i, \sigma(v))) \right ),$$
$$ b_i =  \exp\left (-\beta\sum_{u\in \mathcal{N}(w)}\mathbb{V}(i, \tau(u)) \right ) = \exp\left (-\beta (S_i + \mathbb{V}(i, \tau(v))) \right ).$$
Thus,
$$
\begin{cases}
p_i(\sigma, w) = \frac{a_i}{a_1 + \ ...\ a_{m}} \\
p_i(\tau, w) = \frac{b_i}{b_1 + \ ...\ + b_{m}}
\end{cases}.
$$
The following inequality will be useful:
\begin{equation}
\frac{a_i b_k}{a_k b_i} = \exp(-\beta(\mathbb{V}(i, \sigma(v)) + \mathbb{V}(k, \tau(v)) - \mathbb{V}(k, \sigma(v)) - \mathbb{V}(i, \tau(v))) \leq e^{4\beta|\mathbb{V}|}.
\end{equation}
Then, the upper bound for (3) can be derived as follows:
$$\sum_{k = 1}^m| pref_{k}(\sigma, w) - pref_{k}(\tau ,w) | \leq \sum_{k=1}^{m}\sum_{i = 1}^k|p_i(\sigma, w) - p_i(\tau, w)| \leq $$
$$\leq m \sum_{i = 1}^m|p_i(\sigma, w) - p_i(\tau, w)| =  m\sum_{i=1}^m\left | \frac{a_i}{a_1 +...+a_{m}} - \frac{b_i}{b_1 +...+ b_{m}}\right| \leq$$
$$ \leq \frac{m}{(a_1+...+a_{m})(b_1+...+b_{m})}\sum_{i = 1}^m |a_i(b_1 +...+b_m) - b_i(a_1 + ...+a_m)| \leq $$ $$\leq \frac{m}{(a_1+...+a_{m})(b_1+...+b_{m})}\sum_{i=1}^m \sum_{j = 1}^m |a_i b_j - a_j b_i|  \leq 
$$
\begin{equation}
\leq \frac{m}{(a_1+...+a_{m})(b_1+...+b_{m})}\sum_{i=1}^m \sum_{j = 1}^m a_j b_i \left 
|e^{4\beta|\mathbb{V}|} - 1 \right | \leq 
m\left(e^{4\beta |\mathbb{V}|} - 1\right).
\end{equation}
And now collecting together (2), (3) and (5), we obtain
\begin{equation}
{\sf E}\rho(X_{\sigma}^1, X_{\tau}^1) \leq  1 - \frac{1 - \triangle m e^{4\beta|\mathbb{V}|}}{n} 
\leq \exp\left (-\frac{1 - \triangle m(e^{4\beta |\mathbb{V}|} - 1)}{n} \right ).
\end{equation}
Indeed, the diameter of $\Omega$ is equal to $n(m-1)$ and it corresponds to the distance between configurations $\hat{1} = (1,1, \ ...\ ,1)$ and $\hat{m} = (m,m,\ ...\  ,m)$. Now invoking Lemma~5 with $\alpha$ provided by (6), we obtain the upper bound
for $t_{mix}(\varepsilon)$ given in the theorem statement.
\hfill $\Box$
\end{proof}

Once we proved the theorem, we can think about modifications of the interaction matrix $\mathbb{V}$ and their influence on the model. It is easy to see from the definition of the Gibbs distribution that if we consider matrix $c\mathbb{V}$, where each element of matrix $\mathbb{V}$ is multiplied by a factor $c$, 
we obtain a new probability distribution on the configuration space $\Omega$ which is actually equal to the Gibbs distribution for the pair $\mathbb{V}$ and $c\cdot\beta$. Moreover, if we
add some constant $d$ to all elements of matrix $\mathbb{V}$, then the distribution will not change at all. 
Now we notice that $|\mathbb{V}|$ is mentioned in Theorem~1 and
we can diminish it to some extent. This results in the following
refinement.

\begin{col}
Let $\triangle$ be the maximum degree of graph $G = (V,E),\ |V| = n$ and $\mathbb{V}$ be the interaction matrix. Let also $\beta$ be the inverse temperature and $M = \{1,2,\ ...\ ,m\}$ be the set of attribute values. Let also $$K =  \frac{\max\limits_{x,y}\mathbb{V}(x,y) - \min\limits_{x,y}\mathbb{V}(x,y)}{2}.$$ 
If
$$\beta < \frac{1}{4K}\ln \left(1 + \frac{1}{\triangle m}\right), $$
then
$$ t_{mix} \leq \left  \lceil \frac{n(\ln(n)+ \ln(m-1) + \ln(\frac{1}{\varepsilon}))}{1 - \triangle m(e^{4\beta K} - 1)}\right \rceil. $$
\end{col}

This refinement gives a slightly better bound for the mixing time.
However, we prefer to keep both formulations since the first variant could be just more notationally convenient in some setting.

In the case of quadratic dependencies in $\mathbb{V}$ we obtain even better upper bound.
\begin{Theorem}
If $\ \mathbb{V}(x,y) = (x - y)^2$, and
$$ \beta < \frac{1}{2(m-1)}\ln\left(1 + \frac{1}{\triangle m}\right),  $$
then 
$$ t_{mix} \leq \left  \lceil \frac{n(\ln(n)+ \ln(m-1) + \ln(\frac{1}{\varepsilon}))}{1 - \triangle m(e^{2\beta (m-1)} - 1)}\right \rceil. $$
\end{Theorem}
In this particular case $|\mathbb{V}| = (m-1)^2$ and the above
mentioned result is obviously more efficient than the one which 
can be obtained from Corollary~1.

\begin{proof}
The only difference in the proof of this theorem with respect to the previous results is in inequality (4). Recall that we use that inequality only for neighbour configurations $\sigma$ and $\tau$, which means that there is a vertex $v$ such that $\sigma$ and $\tau$ agree everywhere but in vertex $v$, and for that vertex it holds that
$|\sigma(v) - \tau(v)| = 1$. Since  $\mathbb{V}(x,y) = (x - y)^2$, we can rewrite the right hand side of inequality (4) in the following way:
 $$\frac{a_i b_k}{a_k b_i} = \exp(-\beta((i-\sigma(v))^2 + (k - \tau(v))^2 - (k - \sigma(v))^2 - (k - \tau(v))^2)),$$
Now, without loss of generality $\sigma(v) +1 = \tau(v)$, and then
\begin{equation}
\frac{a_i b_k}{a_k b_i} = \exp(2\beta(k-i)) \leq \exp(2\beta (m-1)).
\end{equation}
The latter provides us $\alpha$ for Lemma~5 and leads to 
the proof of the theorem. 
\hfill $\Box$
\end{proof}

\noindent
\textbf{Remark} All three results mentioned above show that there is fast mixing with respect to some condition on the temperature of the system. Actually, it is impossible to proof fast mixing in general case independently of the temperature. It is already shown for the Ising model, and we can generalize that fact and can demonstrate that for arbitrary $m$ and $m\times m$ matrix $\mathbb{V}$, where not all elements are equal, there exists a temperature and a graph such that mixing time has exponential order in terms of graph size. Moreover, 
we believe, that for every $m$ and $V$ there exists an example of a graph such that mixing is fast independently of the temperature.
This is a good question to address in future research.

\section{Simulations}

\subsection{Monotone perfect Markov Chain Monte Carlo}
In this section we are about to compare theoretical result with real simulations. Of course, for simulation one can just run the Glauber
dynamics and use the bounds on the mixing time from Theorem~1 or
Corollary~1 to indicate the simulation stopping time. However,
if matrix $\mathbb{V}$ has some structure, it appears to be possible
to construct a monotone perfect Markov Chain Monte Carlo (MCMC) simulation which produces perfect sampling and has a natural stopping
rule. Our construction is based on the general recommendations
given in \cite{ProppWilson}. Towards this end, under coupling
described by equation (1), we need to show that for any two configurations $\sigma$ and $\tau$, such that $\sigma \preceq \tau$,
we have $X_{\sigma}^t(U,w) \preceq X_{\tau}^t(U,w)$, where the order $\preceq$ means that for all vertices $v \in V$ it holds that $\sigma(v) \leq \tau(v)$. Unfortunately, this is true not for any matrix $\mathbb{V}$
and here, unlike in Theorem~1, we have to impose additional restrictions on $\mathbb{V}$.

Let us call matrix $\mathbb{V}$ \textit{submodular} if for all $i<j , k<l$ it holds that
$$
\mathbb{V}(i,k)+\mathbb{V}(j,l) \leq \mathbb{V}(i,l)+\mathbb{V}(j,k).
$$
For example, matrix $\mathbb{V}(x,y) = f(x-y)$ is \textit{submodular}, when $f$ is a convex function (in particular, the matrix $\mathbb{V}$ in Theorem~2 is submodular).

\begin{Lemma} Let $\sigma \preceq \tau$ and there is a coupling defined by equality (1) for submodular matrix $\mathbb{V}$. Then  $$
X_{\sigma}^t(U,w) \preceq X_{\tau}^t(U,w).
$$
\end{Lemma}
\begin{proof}
Suppose $t = 1$. Since the introduced order is transitive, we can
limit consideration to neighbor configurations. So, let $\sigma(u) = \tau(u)$ for all $u \in V\setminus \{v\}$ and $\sigma(v) +1= \tau(v)$. Let some vertex $w$ be chosen for update. If $w \notin \mathcal{N}(v)$ then the neighborhood of $w$ is the same for both configurations and it holds that $X_{\sigma}^1(U,w)(w) = X_{\tau}^1(U,w)(w)$. Then, consider $w \in N(v)$. It will be enough to prove that for all $k \leq m$ the following
inequality holds
$$
pref_{k}(\sigma, w) \leq pref_{k}(\tau ,w)
$$
to be sure that
$$
X_{\sigma}^1(U,w)(w) = \min(k | pref_k(\sigma,w) \geq U) \leq  \min(k | pref_k(\tau,w) \geq U) = X_{\tau}^1(U,w)(w).
$$
Here we will use notations of Lemma~6.
$$
pref_{k}(\sigma, w) - pref_{k}(\tau ,w)  = \sum_{i = 0}^k p_i(\sigma, w) - \sum_{i = 0}^k p_i(\tau, w) =
$$
$$ \sum_{i=0}^k  \frac{a_i}{a_0 +...+a_{m}} - \sum_{i=0}^k\frac{b_i}{b_0 +...+ b_{m}}  =$$

$$= \frac{ (a_0 +...+a_{k}) \cdot (b_0 +...+ b_{m}) -  (a_0 +...+a_{m}) \cdot (b_0 +...+ b_{k}) }{(a_0 +...+a_{m})(b_0 +...+ b_{m})} =$$

$$ = \frac{ (a_0 +...+a_{k}) \cdot (b_{k+1} +...+ b_{m}) -  (a_{k+1} +...+a_{m}) \cdot (b_0 +...+ b_{k}) }{(a_0 +...+a_{m})(b_0 +...+ b_{m})} = $$

$$ \frac{1}{(a_0+...+a_{m})(b_0+...+b_{m})}\sum_{i\leq k <j}^m( a_i b_j - a_j b_i ) \leq 0.$$
The last inequality holds since each summand is at most zero: it is provided by equation (4), submodular property of matrix $\mathbb{V}$ and the fact that summation is performed with $i<j$.  By induction argument
the proof immediately extends for arbitrary $t$. \hfill $\Box$
\end{proof}

Now we can propose the following algorithm:

\begin{algorithm}
\caption{Monotone perfect MCMC}
\begin{algorithmic}
\State $U_t \gets \text{random uniform variables from the segment [0,1]}$
\State $w_t \gets \text{random uniform variables from the set $V$}$
\State $T \gets 1$
\Repeat
\State $upper\gets \hat{m}$
\State $lower\gets \hat{1}$

\For{$t = -T \ldots -1$}
\State $upper \gets X^1_{upper}(U_t,w_t)$
\State $lower \gets X^1_{lower}(U_t,w_t)$
\EndFor
\State $T \gets 2T$
\Until{$upper = lower$}
\State \Return $upper,T$
\end{algorithmic}
\end{algorithm}
It is needed to say that the algorithm uses the same random pair $(U_t,w_t)$ at the same $t$, that is why we initialize them only once during the first call.
The required number of steps for this algorithm is upper bounded by $4T_*$, where $T_*$ is the smallest T such that $upper$ and $lower$ values converge. In this case $T_*$ is a random value depending on $U_t$ and $w_t$. Having found $T$ such that $T<T_* \leq 2T$ one can make a binary search to find out the accurate value of $T_*$. This calculation has asymptotic complexity of order $T_*\ln T_*$.

According to \cite{ProppWilson}, we have:
$${\sf E} T_* \leq 2t_{mix}\cdot(1+\ln n + \ln m).$$
This gives an idea that the Glauber dynamics and Monotone perfect
MCMC are comparable in terms of computational requirements.
Of course, the advantage of the monotone perfect MCMC is that
it produces sampling from the exact stationary distribution.

\subsection{Numerical example with real network}
Let consider well-known social network with attributes {\it AddHealth} \cite{AddHealth1}. For our experiments, we take as attribute 
the grade (class) of a pupil at school. It is an ordinal attribute
in the interval between 7 and 12. It seems natural that this network has cluster structure based on class attribute, because the probability of friendship between two pupils is bigger if their classes are not so far
apart in time. For this purpose, as in \cite{ANT16}, we have chosen $6\times 6$ interaction matrix $\mathbb{V}(x,y) = (x - y)^2$. Since $\mathbb{V}$ is submodular, we can use monotone perfect MCMC. 
We have taken publically available AddHealth graph \cite{AddHealth2} with the number of vertices $n = 1996$ and with the maximum degree $\triangle = 36$. In this case Theorem~2 provides fast mixing for $\beta < 0.000461895$, or equivalently, for the temperature $> 2165$.

If we choose $\beta = 0.0002$, Theorem~2 gives the upper bound $27000$ on the mixing time while perfect MCMC algorithm makes about $20000-25000$ running steps. Moreover, if we choose $\beta$ bigger than provided by Theorem~2, e.g., about $0.04$, the perfect MCMC is still fast enough 
finishing approximately after $200000$ steps. Since we have a relation between the expectation of the number of steps in perfect MCMC and 
the mixing time, we realize that, on the one hand, our theorem is in agreement with experiment and, on the other hand, on that particular graph there is fast mixing on broader set of parameters. The question if it is possible to obtain a tighter mixing time estimate is an interesting direction for future research.

We have also tried to fit the value of $\beta$ for the AddHealth data using a variation of the method of moments (see e.g., \cite{S01}). Specifically, we tried to fit the simulated energy to the energy of the AddHealth data, which is equal to 12328. The perfect simulation algorithm converges in acceptable time for $\beta$ as low as 0.125,
which gives the energy level around 15000. We think it is a reasonable match. It is interesting that AddHealth social network is on the boundary of rapid mixing. This might not be a coincidence as a social network can self-organize to find a balance between sufficiently rapid mixing and division into communities. 


\end{document}